\documentclass[a4paper,UKenglish,cleveref, autoref]{lipics-v2019}

\usepackage[utf8]{inputenc}

\usepackage{xspace}
\usepackage{amsmath}
\usepackage{amsthm}
\usepackage{soul}
\usepackage[table]{xcolor}
\usepackage{xspace}

\newtheorem{problem}{\textbf{Problem}}

\hideLIPIcs

\newcommand\B{\ensuremath{\mathtt{b}}\xspace}
\newcommand\E{\ensuremath{\mathtt{e}}\xspace}

\newcommand\BB{\ensuremath{\mathtt{B}}\xspace}
\renewcommand\AA{\ensuremath{\mathtt{A}}\xspace}
\newcommand\X{\ensuremath{\mathtt{s}}\xspace}
\newcommand\Y{\ensuremath{\mathtt{t}}\xspace}
\newcommand\zero{\ensuremath{\mathtt{0}}\xspace}
\newcommand\one{\ensuremath{\mathtt{1}}\xspace}
\newcommand\zeroAB{\AA\BB\AA\BB\AA\BB\AA}
\newcommand\oneAB{\AA\BB\AA}

\newcommand\pmlg{\textsf{PMLG}\xspace}
\newcommand\ov{\textsf{OV}\xspace}

\newcommand\seth{\textsf{SETH}\xspace}

\title{{On the Complexity of Exact Pattern Matching in Graphs: Determinism and Zig-Zag Matching}} 

\titlerunning{On the Complexity of Exact Pattern Matching in Graphs}

\author{Massimo Equi}{Department of Computer Science, University of Helsinki, Finland}{massimo.equi@helsinki.fi}{}{}

\author{Roberto Grossi}{Dipartimento di Informatica, Universit\`a di Pisa, Italy}{grossi@di.unipi.it}{}{}

\author{Veli M\"{a}kinen}{Department of Computer Science, University of Helsinki, Finland}{veli.makinen@helsinki.fi}{}{}

\author{Alexandru I. Tomescu}{Department of Computer Science, University of Helsinki, Finland}{alexandru.tomescu@helsinki.fi}{}{}

\authorrunning{M. Equi, R. Grossi, V. M\"{a}kinen, A. I. Tomescu}

\Copyright{M. Equi, R. Grossi, V. M\"{a}kinen, A. I. Tomescu}



\category{}

\relatedversion{}

\supplement{}
\funding{This work has been partially supported by Academy of Finland (grant 309048)}

\nolinenumbers 

\EventEditors{John Q. Open and Joan R. Access}
\EventNoEds{2}
\EventLongTitle{42nd Conference on Very Important Topics (CVIT 2016)}
\EventShortTitle{CVIT 2016}
\EventAcronym{CVIT}
\EventYear{2016}
\EventDate{December 24--27, 2016}
\EventLocation{Little Whinging, United Kingdom}
\EventLogo{}
\SeriesVolume{42}
\ArticleNo{23}

\begin{document}

\maketitle

\begin{abstract}
Exact pattern matching in labeled graphs is the problem of searching paths of a graph $G=(V,E)$ that spell the same string as the given pattern $P[1..m]$. This basic problem can be found at the heart of more complex operations on variation graphs in computational biology, query operations in graph databases, and analysis of heterogeneous networks, where the nodes of some paths must match a sequence of labels or types. In our recent work we described a conditional lower bound stating that the exact pattern matching problem in labeled graphs cannot be solved in less than quadratic time, namely, $O(|E|^{1 - \epsilon} \, m)$ time or $O(|E| \, m^{1 - \epsilon})$ time for any constant $\epsilon>0$, unless the Strong Exponential Time Hypothesis (\seth) is false. The result holds even if node labels and pattern $P$ are drawn from a binary alphabet, and $G$ is restricted to undirected graphs of maximum degree three or directed acyclic graphs of maximum sum of indegree and outdegree three. It was left open what happens on undirected graphs of maximum degree two, i.e., when the pattern can have a \emph{zig-zag} match in a (cyclic) bidirectional string. Also, the reduction created a non-determistic directed acyclic graph, and it was left open if \emph{determinism} would make the problem easier. In this work, we show through the \emph{Orthogonal Vectors} hypothesis (\ov) that the same conditional lower bound holds even for these restricted cases.
\end{abstract}

\section{Introduction}
\label{sec:typesetting-summary}

Large-scale labeled graphs are becoming ubiquitous in several areas, such as computational biology, graph databases, and graph mining. Applications require sophisticated operations on these graphs, and often rely on primitives that locate paths whose nodes have labels or types matching a pattern given at query time. We refer the reader to the introduction in Equi et al.~\cite{EGM19} for a thorough review of these applications and references therein.

The \emph{Pattern Matching in Labeled Graphs} (\pmlg) problem is as follows.
Given a labeled graph $G=(V,E,L)$ and a pattern string $P$ over the alphabet $\Sigma$, where $L : \rightarrow \Sigma^+$ represents the labeling of the nodes with strings, we say that $P$ occurs in $G$ if there is a sequence of adjacent (and not necessarily distinct) nodes $u_1, \ldots, u_k$ (where $(u_i, u_{i+1}) \in E$ for $1 \leq i < k$) such that $P$ occurs as a substring in the string obtained by their label concatenation $L(u_1) \cdots L(u_k)$. 
Given a pattern $P$ and a labeled graph $G$ over $\Sigma$,
the exact pattern matching problem on graphs requires to establish if $P$ occurs in $G$. Notice that allowing matching nodes not to be distinct is an essential property to make our problem of interest with respect to conditional hardness: With the requirement of distinctness, one can easily derive NP-hardness from the Hamiltonian path problem with a pattern consisting of $a^{|V|}$ and each node labeled with $a$. Generalization of such reductions are considered in ~\cite{LCRP16}.   

Although there is a growing need to perform pattern matching on graphs, the idea of extending the problem of string searching in sequences to pattern matching in graphs was studied over 25 years ago as a search problem in \emph{hypertext} \cite{manber1992approximate}. Interestingly, the best bounds achieved for both exact pattern matching \cite{AmirLL00} and approximate pattern matching \cite{rautiainen2017aligning} in graphs are the same: both solutions take $O(N + m|E|)$ time, where $N = \sum_{u \in V} |L(u)|$ is the total length of text strings in all nodes, $m=|P|$ is the pattern length, and $|E|$ is the number of edges in $G$.

The quadratic cost of the approximate matching in graphs by Rautiainen and Marschall~\cite{rautiainen2017aligning} is asymptotically optimal under the Strong Exponential Time Hypothesis~\cite{IP01} (\seth) as (i)~they solve the approximate string matching as a special case, since a graph consisting of just one path of $|E|+1$ nodes and $|E|$ edges is a text string of length $n=|E|+1$, and (ii)~it has been recently proved that the edit distance of two strings of length $n$ cannot be computed in $O(n^{2-\epsilon})$ time, for any constant $\epsilon>0$, unless the \emph{Strong Exponential Time Hypothesis} (\seth) is false~\cite{BI15}. Hence this conditional lower bound explains why the $O(m|E|)$ barrier has been difficult to cross. 

In our recent work \cite{EGM19}, we showed that \emph{exact and approximate pattern matching are equally hard on graphs} under \seth. Namely, we showed the conditional lower bound that an $O(|E|^{1 - \epsilon} \, m)$-time or an $O(|E| \, m^{1 - \epsilon})$-time algorithm for exact pattern matching on graphs cannot be achieved unless \seth is false. This result explains why it has been difficult to find indexing schemes for graphs with other than best case or average case guarantees for fast exact pattern matching \cite{SVM14,GMS17}.

Our reduction \cite{EGM19} covered a wide range of graphs and patterns: the reduction graph has maximum degree three, string and node labels are drawn from binary alphabet, and the graph is non-deterministic when interpreted as acyclic directed graph (DAG). While some special cases are already covered in literature (see conclusions in our earlier work \cite{EGM19}), there were two important cases left to complete the picture: 1) \emph{zig-zag matching} in a (cyclic) bidirectional string (undirected graph of maximum degree two), and 2) matching on \emph{deterministic directed acyclic graphs}. In this work, we give two new reductions that show that the same conditional lower bound holds even in these restricted cases.

Like our previous reduction \cite{EGM19}, our new reductions share some similarities with those for string problems \cite{BI15,BK15,ABW15,BI16,BZ17}; now we also inherit the use of the \emph{Orthogonal Vectors} hypothesis (\ov) to establish the connection with \seth, while earlier \cite{EGM19} we gave a direct reduction from \seth. The closest connection is with a conditional hardness on several forms of regular expression matching \cite{BI16}. Our reduction has additional constructions to allow a pattern to match inside any type of graphs, while a \emph{non-deterministic finite automaton} (NFA) built on top of those specific regular expressions would accept only certain sub-patterns. For zig-zag matching the reduction is more intricate as the underlying graph has less structure. To cover the deterministic case, the new reduction has a special design to enable local merging of graph substructures that avoids exponential growth. This part of the reduction is especially interesting, as converting an NFA into a \emph{deterministic finite automaton} (DFA) can take exponential time \cite{RS59}, but this gap appears to have no significance with regards to hardness of exact pattern matching. This result complements recent findings about \emph{Wheeler graphs} \cite{GMS17,GT19,APP19}. Wheeler graphs are a class of graphs that admit an index structure supporting linear time exact pattern matching. Gibney and Thankachan \cite{GT19} show that it is NP-complete to recognize whether a (non-deterministic) DAG is a Wheeler graph. Alanko, Policriti, and Prezza \cite{APP19} give a linear time algorithm for recognizing whether a deterministic DAG is a Wheeler graph. Our result shows that converting an arbitrary deterministic DAG into an equivalent Wheeler graph should take quadratic time unless \seth fails.

\section{Definitions}
In this work we present two reductions from the \emph{Orthogonal Vectors} (\ov) problem to the two special cases, deterministic and zig-zag, of the pattern matching in labeled graphs (\pmlg) problem. These reductions will eventually let us conclude that it is not possible to solve \pmlg even in these special cases either in $O(|E|^{1-\epsilon}m)$ or $O(|E|m^{1-\epsilon})$ without contradicting the \ov Hypothesis, complementing the earlier results \cite{EGM19}. Contradiction with \ov Hypotheses implies contradiction with \seth~\cite{Wil05}.  

Before giving the formal definitions of the \pmlg problem, we recall the \ov Hypothesis. In the \ov problem we are given two sets $X, Y \subseteq \{ 0,1 \}^d$ such that $|X| = |Y| = n$ and $d = \omega(\log n)$, and we are asked to determine whether or not there exist $x \in X$ and $y \in Y$ such that $x \cdot y = 0$, where $x \cdot y = \Sigma_{i=0}^d \; x[i] \cdot y[i]$. The \ov Hypothesis~\cite{Wil05} states that for any constant $\epsilon > 0$, no algorithm can solve \ov in $O(n^{2-\epsilon}\text{poly}(d))$ time.
We can consider now the \pmlg problem. 

\begin{definition}[Labeled graph]
\label{definition:labeledgraph}
    Given an alphabet $\Sigma$, a \emph{labeled graph} $G$ is a triplet $(V,E,L)$ where $(V,E)$ is a directed or undirected graph and $L: V \mapsto \Sigma^+$ is a function that defines which string (i.e., \emph{label}) over $\Sigma$ is assigned to each node.
    
    If $G$ is a directed graph, then we say that $G$ is \emph{deterministic} if for any node, it holds that the first symbol in the labels of any two out-neighbors is always different.
    
    A node labeled with a character $\sigma \in \Sigma$ is called a \emph{$\sigma$-node}. And edge whose endpoints have labels $\sigma_1$ and $\sigma_2$, respectively, is called a \emph{$\sigma_1\sigma_2$-edge}.
\end{definition}

\begin{definition}[Match]
    Let $u_1, \ldots, u_j$ be a path in a labeled graph $G$ (paths are allowed to repeat nodes) and let $P$ be a pattern. Also, let $L(u)[l:]$ and $L(u)[:l']$ denote the suffix of $L(u)$ starting at position $l$, and the prefix of $L(u)$ ending at position $l'$, respectively. We say that $u_1, \ldots, u_j$ is a \emph{match} for $P$ in $G$ with offset $l$ if the concatenation of the strings $L(u_1)[l:] \cdot L(u_2) \cdot \ldots \cdot L(u_{j-1}) \cdot L(u_j)[:l']$ equals $P$, for some $l'$.
\end{definition}

\begin{problem}[Pattern Matching in Labeled Graphs (\pmlg)]
\item{\textsc{input}:} A labeled graph $G = (V,E,L)$ and a pattern $P$, both over an alphabet $\Sigma$.
\item{\textsc{output}:} All the matches for $P$ in $G$.
\end{problem}

Before presenting our main result, we give a new proof regarding the general \pmlg problem. 
\begin{theorem} \label{theorem:Emlowerbound}
For any constant $\epsilon > 0$, the Pattern Matching in Labeled Graphs (\pmlg) problem for a binary alphabet cannot be solved in either $O(|E|^{1-\epsilon} \, m)$ or $O(|E| \, m^{1-\epsilon})$ time unless the \ov Hypothesis fails. 
\end{theorem}

An analogous statement was proven in our earlier work \cite{EGM19} with respect to \seth. Although that reduction \cite{EGM19} can easily be modified to use \ov instead, converting it to cover the deterministic case appears not to be easy. Therefore we give a new reduction that can be strengthened to yield the following result. 

\begin{theorem}
\label{cor:deterministic-dag}
The conditional lower bound stated in Theorem \ref{theorem:Emlowerbound} holds even if it is restricted to labeled \emph{deterministic} directed acyclic graphs (DAGs) with maximum sum of outdegree and indegree three.
\end{theorem}

The new reduction does not cover the other special case --- zig-zag matching in bidirectional string. For this we give another reduction and prove the following result.

\begin{theorem}
\label{thm:zig-zag}
The conditional lower bound stated in Theorem \ref{theorem:Emlowerbound} holds even if it is restricted to undirected graphs of maximum degree two, where the pattern and the node labels are drawn from an alphabet of size at least 6. 
\end{theorem}

\section{Undirected Graphs}
\label{sec:reduction}
Given an instance of \ov with sets $X$ and $Y$, we will construct a pattern $P$ and a graph $G$ such that $P$ will have a match in $G$ if and only if there exist a pair of orthogonal vectors between $X$ and $Y$. We first describe how to build the pattern $P$ and then how to obtain the graph $G$.

\subsection{Pattern}
Pattern $P$ is defined over the alphabet $\Sigma = \{ \B,\E,\zero,\one \}$, using the first set of vectors $X = \{ x_1, \ldots, x_n\}$, as the concatenation $P = \B\B P_{x_1}\E\,\B P_{x_2}\E \ldots \B P_{x_n}\E\E$, where for all $h \in \{1,\dots,d\}$, the $h$-th symbol of string $P_{x_i}$ is either \zero or \one, such that $P_{x_i}[h] = \one$ if and only of $x_i[h] = 1$.
We thus view the vectors in $X$ as sub-patterns which are concatenated placing separator characters. Note that the pattern starts with two consecutive characters \B\ and ends with two consecutive characters \E. Such strings are found nowhere else in the pattern, thus they mark its beginning and its end. 


\subsection{Graph}
\label{sub:gadgets}

\begin{sloppypar}
The gadget implementing the main logic of the reduction is an undirected graph $G_W = (V_W,E_W,L_W)$, illustrated in \figurename~\ref{fig:G_W} and defined as follows, starting from the second set of vectors $Y = \{ y_1, \ldots, y_{n}\}$.
\begin{figure}[ht]
    \centering
    \vspace{20pt}
    \par \noindent
    \includegraphics[scale=.60]{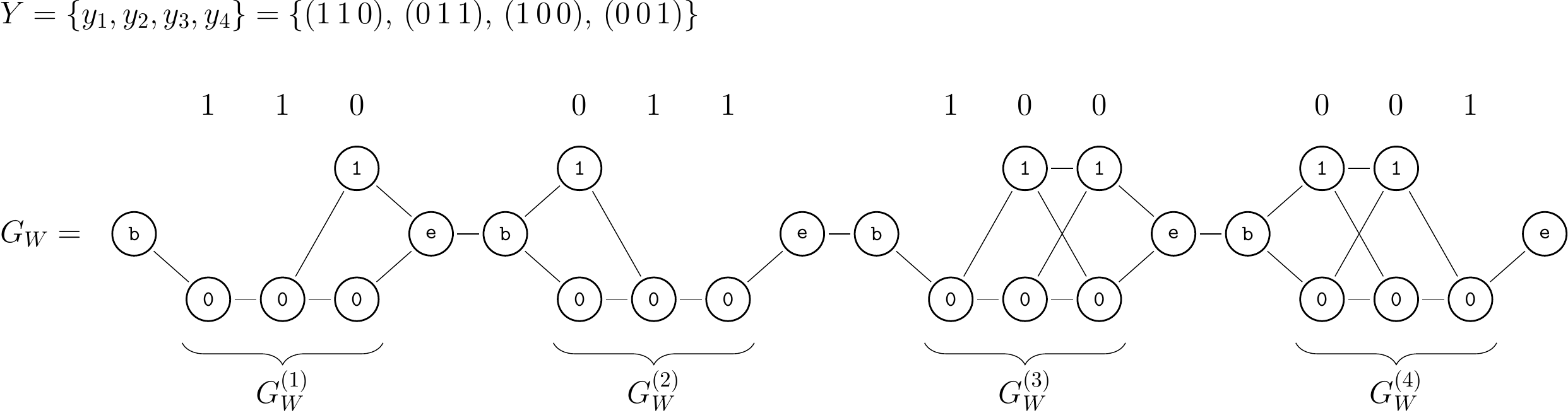}
    \par
    \vspace{20pt}
    \caption{ Gadget $G_W$. }
    \label{fig:G_W}
\end{figure} Set $V_W$ can be seen as $n$ conceptual \emph{groups} of nodes $V_W^{(1)}, V_W^{(2)}, \ldots, V_W^{(n)}$, one for each vector in $Y$, organized in $n$ sub-structures $G_W^{(1)} = (V_W^{(1)}, E_W^{(1)})$, \ldots, $G_W^{(n)} = (V_W^{(n)}, E_W^{(n)})$. See Figure~\ref{fig:G_W} for an example. 
\end{sloppypar}

The idea is to make $G_W^{(j)}$ accept some sub-pattern $P_{x_i}$ if and only if $x_i \cdot y_j = 0$. When building $V_W^{(j)}$ we scan every entry $y_j[h]$ of vector $y_j \in Y$, for $1 \leq h \leq d$. If $y_j[h] = 1$ then vector $x_i$ must be accepted if and only if $P_{x_i}[h] = \zero$, hence we place only a $\zero$-node $v^0_{j h}$ (i.e., labeled with $\zero$). Instead, if $y_j[h] = 0$, the value of vector $x_i$ in position $h$ does not matter and we place both a $\zero$-node $v^0_{j h}$ and a $\one$-node $v^1_{j h}$ (i.e., labeled with \zero\ and \one, respectively).
For each $y_j \in Y$, set $V_W$ also contains some special nodes: a begin $\B$-node $b_W^{(j)}$ and an end $\E$-node $e_W^{(j)}$ (i.e., labeled with \B and \E, respectively). Formally, we have:
\begin{align*}
V_W^{(j)} =& \{ v^0_{j h} \, \mid \,  y_j \in Y, \, 1 \leq h \leq d \} \cup \{ v^1_{j h} \, \mid \, y_j[h] = 0, \; y_j \in Y, \, 1 \leq h \leq d \},\\
\\
V_W =& \bigcup\limits_{j=1}^n \left( V_W^{(j)} \cup \{ b_W^{(j)}, e_W^{(j)} \} \right).
\end{align*}

The edges in set $E_W$ properly connect both the nodes inside each group $V_W^{(j)}$ and the groups themselves with each other. In group $V_W^{(j)}$, node $b_W^{(j)}$ is connected with $v^0_{j 1}$ and, if present, with $v^1_{j 1}$. Also, we place edges connecting nodes $v^0_{j d}$ and $v^1_{j d}$ (if present) with node $e_W^{(j)}$.  Moreover, there is an edge for every pair of nodes that share the same $j$ and are consecutive in terms of $h$ coordinate (e.g., $v^1_{j h}, v^0_{j \, h+1}$), for $1 \leq j \leq n$ and $1 \leq h \leq d$. When considering a single $j$, the edges described so far form set $E_W^{(j)}$, namely the edges connecting the nodes in group $V_W^{(j)}$. In this way we obtain sub-structure $G_W^{(j)} = (V_W^{(j)}, E_W^{(j)})$.
We combine such sub-structures together by connecting each $G_W^{(j)}$ to its predecessor via edge $(e_W^{(j-1)}, b_W^{(j)})$, for all $2 \leq j \leq n$. Thus, the final set of edges is:
\begin{align*}
E_W^{(j)} =& \{ \left(v^0_{j h}, v^0_{j \, h+1}\right) \mid v^0_{j h}, v^0_{j \, h+1} \in V_W\} \cup \{ \left(v^0_{j h}, v^1_{j \, h+1}\right) \mid v^0_{j h}, v^1_{j \, h+1} \in V_W \} \, \cup\\
& \{ \left(v^1_{j h}, v^0_{j \, h+1}\right) \mid v^1_{j h}, v^0_{j \, h+1} \in V_W \} \cup \{ \left(v^1_{j h}, v^1_{j \, h+1}\right) \mid v^1_{j h}, v^1_{j \, h+1} \in V_W \},\\
\\
E_W =& \bigcup\limits_{j=1}^n E_W^{(j)}\\
&\cup \{ \left(b_W^{(j)}, v^0_{j 1}\right) \mid b_W^{(j)}, v^0_{j 1} \in V_W \} \cup \{ \left(b_W^{(j)}, v^1_{j 1}\right) \mid b_W^{(j)}, v^1_{j 1} \in V_W \} \\
&\cup \, \{ \left(v^0_{j d}, e_W^{(j)}\right) \mid v^0_{j d}, e_W^{(j)} \in V_W \} \cup \{ \left(v^1_{j d}, e_W^{(j)}\right) \mid v^1_{j d}, e_W^{(j)} \in V_W\} \\
&\cup \, \{ \left(e_W^{(j)}, b_W^{(j+1)}\right) \mid e_W^{(j)}, b_W^{(j+1)} \in V_W \} \; .
\end{align*}

We now observe that pattern occurrences in $G_W$ have some combinatorial properties.

\begin{lemma} 
\label{lemma:samej}
If sub-pattern $\B P_{x_i} \E$ has a match in $G_W$ then
 all the nodes matching $P_{x_i}$ share the same $j$ coordinate and
have distinct and consecutive $h$ coordinates (i.e., $P_{x_i}$ has a match in $G_W^{(j)}$).
\end{lemma}
\begin{proof}
A match for sub-pattern $\B P_{x_i} \E$ must start at node $b_W^{(j)}$ and end at node $e_W^{(j)}$, for some $j$. We can show that such nodes and all the other matching nodes must share the same $j$ coordinate, i.e., they all belong to $G_W^{(j)}$. Indeed, once character \B\ has been matched with node $b_W^{(j)}$, it is not possible to take the path going through node $e_W^{(j-1)}$, since the next character in the sub-pattern is not an \E. Hence, by construction it is not possible to reach any node $e_W^{(j')}$ for $j' < j$. The only way to reach any $e_W^{(j'')}$ such that $j'' > j$ is by following the paths that go through nodes $v^0_{j h}$ or $v^1_{j h}$ for every $1 \leq h \leq d$ and $e_W^{(j)}$. Sub-pattern $\B P_{x_i} \E$ has $d+2$ characters and the shortest path from $b_W^{(j)}$ to $e_W^{(j)}$ is $d+2$ nodes long ($b_W^{(j)}$ and $e_W^{(j)}$ included). Thus, $\B P_{x_i} \E$ has not enough characters to reach any $e_W^{(j'')}$, for $j'' > j$. Moreover, if one node in the path from $b_W^{(j)}$ to $e_W^{(j)}$ is matched twice, sub-pattern $\B P_{x_i} \E$ has not enough characters to reach $e_W^{(j)}$. Finally, if one of the matching nodes were not consecutive in terms of $h$ coordinate, by construction we know that we would not be following the shortest path to $e_W^{(j)}$ hence it would not be possible to complete the match. Therefore, every matching node lies on a simple path from $b_W^{(j)}$ to $e_W^{(j)}$, and has distinct and consecutive $h$ coordinates.
\end{proof}


\begin{lemma} 
\label{lemma:matchingGW}
Sub-pattern $\B P_{x_i}\E$ has a match in $G_W$ if and only if there exist $y_j \in Y$ such that $x_i \cdot y_j = 0$.
\end{lemma}
\begin{proof}
Recall that, by construction, $v^0_{j h} \in V_W^{(j)}$ and $v^1_{j h} \in V_W^{(j)}$ hold for those $h$ such that $y_j[h] = 0$, while $v^0_{j h} \in V_W^{(j)}$ and $v^1_{j h} \not\in V_W^{(j)}$ hold in case $y_j[h] = 1$. We handle the two implications of the statement individually.

($\Rightarrow$) By Lemma~\ref{lemma:samej}, we can focus on the $d$ distinct and consecutive nodes of $G_W^{(j)}$ that match $P_{x_i}$. In particular we know that each character $P_{x_i}[h]$ is matched by either $v^0_{j h}$ or $v^1_{j h}$. Consider vectors $x_i \in X$ and $y_j \in Y$. If $P_{x_i}[h] = \one$ has a match in $G_W^{(j)}$ it means that node $v^1_{j h}$ exists and hence $y_j[h]=0$, implying $x_i[h] \cdot y_j[h] = 0$. If $P_{x_i}[h] = \zero$, by construction we know that $x_i[h]=0$ and, no matter the type of match in $G_W^{(j)}$, it clearly holds that $x_i[h] \cdot y_j[h] = 0$. At this point, we can conclude that $x_i[h] \cdot y_j[h] = 0$ for every $1 \leq h \leq d$, thus $x_i \cdot y_j = 0$.

($\Leftarrow$) Consider vectors $x_i \in X$ and $y_j \in Y$ that are such that $x_i \cdot y_j = 0$. For $h = 1, 2, \ldots, d$, if $y_j[h] = 0$ then $v^0_{j h},v^1_{j h} \in V_W^{(j)}$ and $P_{x_i}[h]$ can match either $v^0_{j h}$ or $v^1_{j h}$ in $G_W^{(j)}$. If $y_j[h] = 1$ it must be $x_i[h] = 0$ since $x_i \cdot y_j = 0$, thus $P_{x_i}[h] = \zero$ and it can match node $v^0_{j h}$, which is always present in $G_W^{(j)}$. Finally, characters \B\ and \E\ can match nodes $b_W^{(j)}$ and $e_W^{(j)}$, respectively. All characters of $\B P_{x_i} \E$ have now a matching node and the definition of the edges in $E_W$ allows to visit all such nodes via a matching path starting at $b_W^{(j)}$ and ending at $e_W^{(j)}$.
\end{proof}

While the previous gadget is useful to check whether a vector $x_i \in X$ is orthogonal to the ones in $Y$ using a \emph{given} sub-pattern $P_{x_i} \in \B \{ \zero, \one \}^d \E$, we need another ``jolly'' gadget that matches \emph{all} sub-patterns in $\B \{ \zero, \one \}^d \E$ (this is useful when $x_i$ is not orthogonal to any $y_j \in Y$).
We will concatenate $2(n-1) = 2n-2$ instances of this ``jolly'' gadget, thus obtaining the graph $G_U = (V_U,E_U,L_U)$ illustrated in \figurename~\ref{fig:G_U}. The reason for using this specific number will be clear when we will describe how to combine this gadget with $G_W$. 

The $j$-th copy of the ``jolly'' gadget $G_U^{(j)}$ has a $\B$-node $b_U^{(j)}$ followed by $\zero$- and \mbox{$\one$-nodes} $v^0_{j h}$, $v^1_{j h}$ and then $\E$-node $e_U^{(j)}$, with $1 \leq j \leq n-1$ and $1 \leq h \leq d$. 
We place the edges $(b_U^{(j)}, v^0_{j 1}), \, (b_j, v^0_{j 1}), \, (v^1_{j d}, e_U^{(j)}), \, (v^1_{j d}, e_U^{(j)})$ for connecting the beginning and ending nodes of each gadget with its inner part.  We connect nodes $v^0_{j h}$ and $v^1_{j h}$ with the edges $(v^0_{j h}, v^0_{j \, h+1}), \, (v^0_{j h}, v^1_{j \, h+1}), \, (v^1_{j h}, v^0_{j \, h+1}), \, (v^1_{j h}, v^1_{j \, h+1})$. In the inner part, we add all edges $(e_U^{(j)}, b_U^{(j+1)})$, for all $j = 1, \ldots, n-1$.
\begin{figure}[ht]
    \centering
    \vspace{20pt}
    \par \noindent
    \includegraphics[scale=.65]{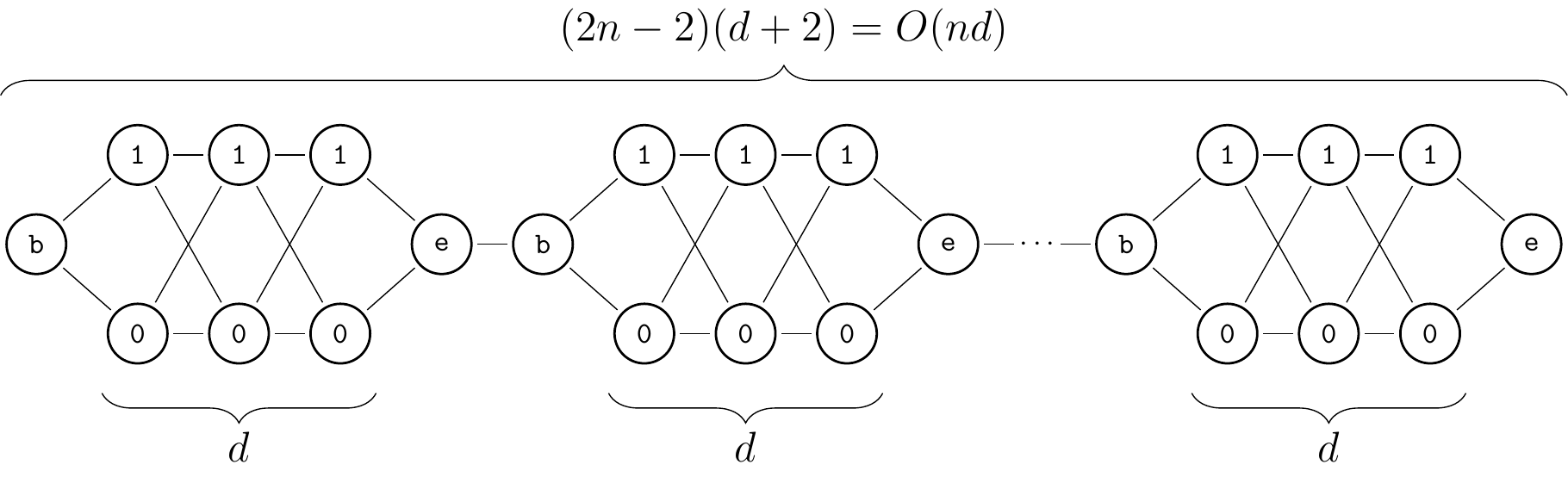}
    \par
    \vspace{20pt}
    \caption{ Gadget $G_U$. }
    \label{fig:G_U}
\end{figure}

\subsection{Putting All Together}
To finish our reduction we use the gadgets described above to build an actual graph $G$ implementing the needed logic. Consider one instance of gadget $G_W$ and two instances of gadgets $G_U$, named $G_{U1}$ and $G_{U2}$. The three gadgets are arranged on three different lines starting with $G_{U1}$ at the top, then $G_W$ in the middle and finally $G_{U2}$ on the bottom line. Sub-structure $G_{U1}^{(n-1+j-1)}$ is placed in such a way to precede $G_W^{(j)}$, which in turn will precede $G_{U2}^{(j)}$, for all $1 \leq j \leq n$. We now connect the three gadgets adding edges $(e_{U1}^{(n-1)}, b_W^{(1)}), \ldots, (e_{U1}^{(2n-2)}, b_W^{(n)})$ and $(e_W^{(1)}, b_{U2}^{(1)}), \ldots, (e_W^{(n)}, b_{U2}^{(n)})$. See Figure~\ref{fig:G}.

The idea is to force the pattern to cross $G_W$ during a match so that it will detect a vector $y_j \in Y$ (represented by $G_W^{(j)}$) orthogonal to some $x_i \in X$ (encoded as $P_{x_i}$), if such a vector exists. As observed earlier, pattern $P$ starts with sequence $\B\B$ and ends with $\E\E$. Moreover, such sequences do not appear anywhere else in the pattern. Thus, we place a new $\B$-node for every $b_{W}^{(j)}$ and we connect these two with an edge as in Figure~\ref{fig:G}, for all $1 \leq j \leq n$. We do the same for $b_{U1}^{(j')}$, for all $1 \leq j' \leq 2n-2$. In the same manner we add new $\E$-nodes for every $e_{W}^{(j)}$ and $e_{U2}^{(j'')}$, for all $1 \leq j \leq n, \, 1 \leq j'' \leq 2n-2$. In this way pattern $P$ can start matching in the graph only in $G_{U1}$ or $G_W$ and end only in $G_W$ or $G_{U2}$. Finally, $2(n-1) = 2n-2$ is the minimum number of instances of sub-gadgets of $G_{U1}$ and $G_{U2}$ required by pattern $P$ to be able to match in $G_W$ any of its sub-patterns (i.e., any sub-pattern from $P_{x_1}$ to $P_{x_n}$).
\begin{figure}[ht]
    \centering
    \vspace{20pt}
    \par \noindent
    \includegraphics[scale=.64]{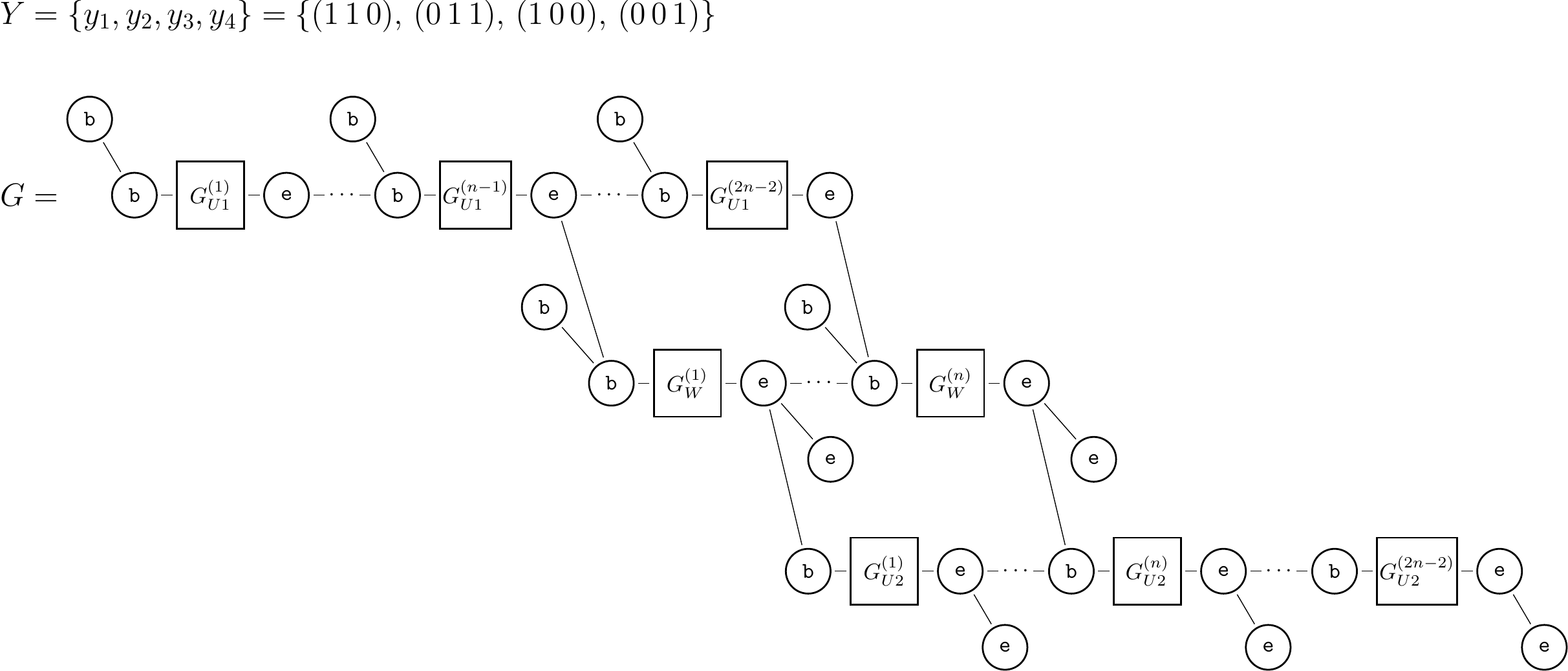}
    \par
    \vspace{20pt}
    \caption{ Final graph $G$. }
    \label{fig:G}
\end{figure}

We now prove that the reduction is correct, first focusing on the sub-patterns of $P$.

\begin{lemma} \label{lemma:patternsubpattern}
Pattern $P$ has a match in $G$ if and only if a sub-pattern $\B P_{x_i}\E$ of $P$ has a match in $G_W$.
\end{lemma}
\begin{proof}
For the $(\Rightarrow)$ implication, the $\E\B$-edges can only be traversed once in this direction, as $P$ contains the sequence $\E\B$ but does not contain $\B\E$. For this reason each distinct sub-pattern $\B P_{x_i}\E$ matches a path from either a distinct portion of $G_{U\ell}$ ($\ell=1,2$) or $G_W$. Moreover, each occurrence of $P$ must begin with $\B\B$ and end with $\E\E$. String $\B\B$ can be matched only in $G_{U1}$ or $G_W$ while $\E\E$ is present only in $G_W$ or $G_{U2}$. Hence, in order to have a full match for pattern $P$ there must exist a sub-pattern $\B P_{x_i}\E$ having a match in $G_W$.

The $(\Leftarrow)$ implication is trivial. In fact, if $\B P_{x_i}\E$ has a match in $G_W$ then we can match $\B P_{x_1}\E \ldots \B P_{x_{i-1}}\E$ in $G_{U1}$ and $\B P_{x_{i+1}}\E \ldots \B P_{x_n}\E$ in $G_{U2}$ by construction, and have a full match for $P$ in $G$.
\end{proof}


We are finally ready to prove Theorem~\ref{theorem:Emlowerbound}.

\begin{proof}[Proof of Theorem~\ref{theorem:Emlowerbound}] First, we prove that the reduction is correct, then we analyze its cost and show how a sub-quadratic time algorithm for \pmlg would contradict the \ov~Hypothesis.

\textbf{Correctness.}
We need to ensure that pattern $P$ has a match in $G$ if and only if there exist vectors $x_i \in X$ and $y_{j} \in Y$ which are orthogonal. This follows from Lemma~\ref{lemma:patternsubpattern}, which guarantees that $P$ has a match in $G$ if and only if a sub-pattern $P_{x_i}$ has a match in $G_W$, and the fact that, by Lemma~\ref{lemma:matchingGW}, this holds if and only if $x_i \cdot y_j = 0$.

\textbf{Alphabet size.}
Our alphabet is of size 4. Analogously to the earlier reduction \cite{EGM19} one can use encoding $\alpha(\zero)=\zero\zero\zero\zero$, $\alpha(\one)=\one\one\one\one$, $\alpha(\B)=\one\zero$, and $\alpha(\E)=\zero\one$ and observe that (after small adjustments considered below) there is a bijection from matches before and after applying the encoding. Nodes with labels of length 2 and 4 can be replaced by chains of nodes labeled by single characters each.  In order to make such encoding work, we need to make the pattern start with characters \E\B\B and end with characters \E\E\B to exploit the properties of the sequence \E\B. Moreover, we have to place and connect a new \E-node to each \B-node used to mark the beginning of a viable match and, in the same manner, we need to add a new \B-node after every \E-node used to mark the end of a match. (See our earlier reduction \cite{EGM19} for a full proof of an analogous adjustment.)

\textbf{Cost.} We analyze the cost of the reduction proving that, given sets of vectors $X$ and $Y$ with $n$ vectors each, the corresponding pattern $P$ and graph $G$ can be built in $O(nd)$ time and space. We observe that each $\B P_{x_i}\E$ in $P$ has $d$ symbols that can be either \zero\ or \one\ plus symbols \B\ and \E. Since $P$ has $n$ sub-patterns $\B P_{x_i}\E$ ad an additional \B\ at the beginning as well as one more \E\ at the end, summing everything up we get a length of $m = n(d+2) + 2 = O(nd)$ symbols. As for $G_U$, it has $2n-2$ sub-gadgets each one having $d$ nodes labeled with \zero, $d$ nodes labeled with \one, and nodes $b_U^{(j)}$ and $e_U^{(j)}$. Hence there are $(2n-2) \, (2d + 2) = O(nd)$ total nodes. Each node has a constant number of incident edges (at most $4$) thus their size is $O(nd)$ as well. The same reasoning applies for $G_W$ as it has $O(nd)$ nodes labeled with \zero\ and at most $O(nd)$ nodes labeled with \one, plus those with \B\ and \E. Also in this case, each node has a constant number of incident edges. For connecting $G_W$ to the two instances of $G_U$ we are adding one edge for every node $b_W^{(j)}$ and $e_W^{(j)}$ in $G_W$, which sum up to $O(nd)$ edges. We are adding one additional $b$ node and edge for each sub-gadget $G_{U1}^{(j)}$ and $G_{W}^{(j)}$ and one additional $e$ node and edge for each sub-gadget $G_{W}^{(j)}$ and $G_{U2}^{(j)}$, which again sum up to $O(nd)$. Since the pattern and the graph have size $O(dn)$, we conclude that the cost of our reduction is indeed $O(nd)$ time and space. The alphabet reduction to size 2 considered above adds a constant factor. 

\textbf{Using the \ov~Hypothesis.}
The last step is to show that any $O(|E|^{1-\epsilon} \, m)$-time or $O(|E| \, m^{1-\epsilon})$-time algorithm $A$ for \pmlg contradicts the \ov Hypothesis. Given two sets of vectors $X$ and $Y$, we can perform our reduction obtaining pattern $P$ and graph $G$ in $O(nd)$ time, by observing that $|E| = O(nd)$ and $m = O(nd)$. No matter whether $A$ has $O(|E|^{1-\epsilon} m)$ or $O(|E| \, m^{1-\epsilon})$ time complexity, we will end up with an algorithm deciding if there exist a pair of orthogonal vectors between $X$ and $Y$ in $O(nd \cdot n^{1-\epsilon}d) = O(n^{2-\epsilon}\text{poly}(d))$ time, which contradicts the \ov Hypothesis.
\end{proof}

\section{Deterministic DAGs}
In this section we show that the undirected graph $G$ obtained from the reduction described in Section~\ref{sec:reduction} can be transformed into a deterministic DAG. We first observe that it is immediate to modify the proof of Theorem~\ref{theorem:Emlowerbound} so that it holds for a DAG. Instead, non-trivial changes have to be applied in order to obtain a deterministic DAG, as stated in Theorem~\ref{cor:deterministic-dag}.

\subsection{DAGs} 

Consider the definitions of edges $E_W$ and $E_U$ in the proof of Theorem \ref{theorem:Emlowerbound}. We immediately obtain a directed acyclic graph
just by considering such edges to be directed. From now onward, when referring to $G$ we will intend its directed version. Indeed, each match of the pattern must begin with a $\B\B$-edge, end with an $\E\E$-edge, and lay along a path between them. So the edges can be oriented by construction from left (first $\B\B$-edge) to right (last $\E\E$-edge) and from top ($G_{U1}$) to bottom ($G_{U2}$). 

\subsection{Determinization}
The directed acyclic graph $G$ obtained in the previous paragraph must be further modified to complete the proof of Theorem~\ref{cor:deterministic-dag}. Looking at $G$ it is clear that the non-deterministic out-neighbours are the $\E$-nodes that are connected with $\B$-nodes in forming the $\E\B$-edges of the graph. For instance, node ${e}_{W}^{(j)}$ has an edge to both ${b}_W^{(j+1)}$ and ${b}_{U2}^{(j)}$. This issue can be handled by merging $G_W$ and $G_{U1}$ into a new gadget $G_{WU}$. See Figure~\ref{fig:G_WU} for an example. 

We start from $G_W$ and add parts of $G_{U1}$ when needed. Consider only sub-structure $G_W^{(j)}$ and assume that the first position in which the $\one$-node is lacking is $h$. We place node $v^1_{jh}$ and the edges $(v^0_{j \, h-1},v^1_{jh})$ and $(v^1_{j \, h-1}, v^1_{jh})$, or edge $({b}_{W}^{(j)},v^1_{jh})$ in case $h=1$. Moreover, we place a partial version of sub-structure ${G}_{U1}^{(j')}, j'=n-1+j$ by adding to the graph the nodes and edges of ${G}_{U1}^{(j')}$ that are located between position $h+1$ and node ${e}_{U1}^{(n-1+j)}$ (included). Now we connect the new node $v^1_{jh}$ to the first $\zero$-node and $\one$-node of such substructure. If there are other positions $h' > h$ lacking the $\one$-node, then $v^1_{j \, h'}$ exists in the partial ${G}_{U1}^{(j')}$ and we connect nodes $v^0_{j \, h'-1}$ and $v^1_{j \, h'-1}$ to it. In case the missing $v^1_{j \, h'}$ is the last one ($h' = d$), we connect $v^0_{j \, h'-1}$ and $v^1_{j \, h'-1}$ to the $\E$-node of the partial ${G}_{U1}^{(j')}$. Finally, we place edge $({e}_{U1}^{(n-1+j)}, {b}_{W}^{(j+1)})$. To conclude the merging, we apply such modification to ${G}_W^{(j)}$ and we remove the edge $({e}_{W}^{(j)}, {b}_{W}^{(j+1)})$ for $1 \leq j \leq n-1$. 

\begin{figure}[ht]
    \centering
    \vspace{20pt}
    \par \noindent
    \includegraphics[scale=.55]{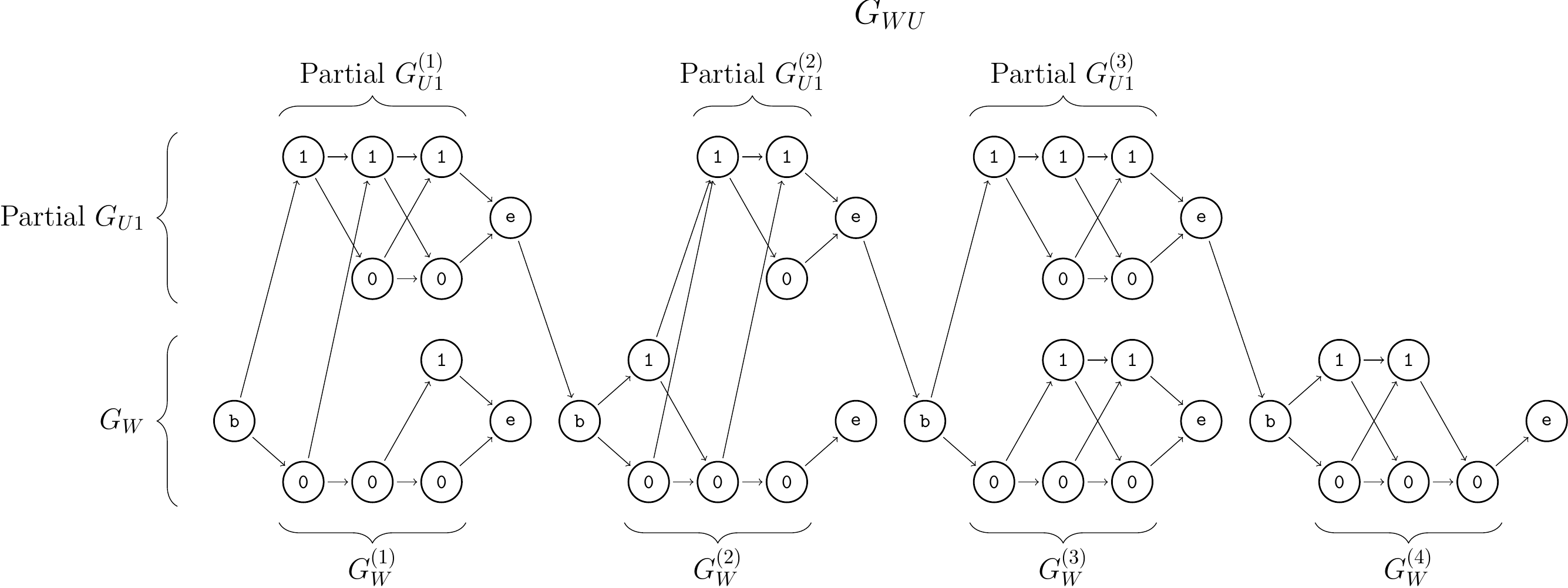}
    \par
    \vspace{20pt}
    \caption{ DAG $G_{WU}$ after merging $G_W$ (from Figure~\ref{fig:G_W}) with $G_{U1}$ (from Figure~\ref{fig:G_U}). }
    \label{fig:G_WU}
\end{figure}

At this point, we place gadget $G_{U2}$ and we connect $G_{WU}$ to it as in $G$. Adding substructures ${G}_{U1}^{(1)}, \ldots, {G}_{U1}^{(n-1)}$ with one additional $\B$-node connected to each of them, and connecting ${G}_{U1}^{(n-1)}$ to ${G}_{W}^{(1)}$ as in $G$, completes the transformation into the new directed acyclic graph, which we call $G'$. Figure~\ref{fig:G'} gives a big picture perspective of such new graph.

Now every node ${e'}_{WU}^{(j)}$ and ${e'}_{U2}^{(j)}$ of graph $G'$ has either only one outgoing edge, or an edge to a $\B$-node and one to an $\E$-node, which means that we removed its only non-deterministic feature.
\begin{figure}[ht]
    \centering
    \vspace{20pt}
    \par \noindent
    \includegraphics[scale=.50]{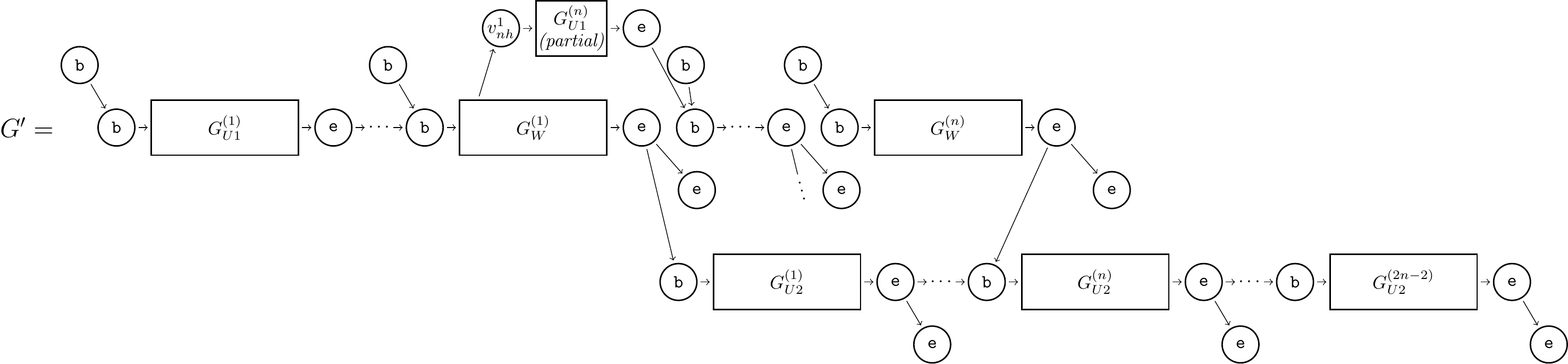}
    \par
    \vspace{20pt}
    \caption{ Final deterministic DAG $G'$. }
    \label{fig:G'}
\end{figure}

\subsection{Correctness and Complexity}
We are now left to prove that the reduction is still correct also for the deterministic DAG $G'$. To this end, we have to adapt the statement of Lemmas \ref{lemma:samej}, \ref{lemma:matchingGW} and \ref{lemma:patternsubpattern} so that they hold for $G'$. The proofs of such lemmas will be almost immediate since they share the most part of the reasoning with the original reduction.

\begin{lemma} \label{lemma:samej_G'}
If sub-pattern $\B P_{x_i} \E$ has a match in $G_{WU}$ then all the nodes matching $P_{x_i}$ share the same $j$ coordinate and have distinct and consecutive $h$ coordinates (i.e., $P_{x_i}$ has a match in $G_W^{(j)}$).
\end{lemma}
\begin{proof}
During the merging process, we added edges for connecting $G_W$ the partial $G_{U1}$ structure and removed the ones between $G_W^{(j)}$ and $G_W^{(j+1)}$, for $1 \leq j \leq n-1$. Nodes $b_W^{(j)}$ and $e_W^{(j)}$ were left connected with $G_W^{(j)}$. Thus, a match can now span one of the new edges and move into the partial $G_{U1}$. In doing so, such match maintains the same $j$ coordinate and the shortest path between node $b_W^{(j)}$ and either node $e_W^{(j)}$ or $e_{U1}^{(j)}$ still is $d+2$ nodes long ($\B$-node and $\E$-node included). Hence, this lemma follows from Lemma \ref{lemma:samej} and the fact that sub-structures ${G}_W^{(j)}$ in $G'$ are the same as in $G$ but for the fact that the edges are directed.
\end{proof}

\begin{lemma} \label{lemma:matchingGW'}
Sub-pattern $\B P_{x_i}\E$ has a match in the underlying sub-structure $G_W$ of $G_{WU}$ if and only if there is $y_j \in Y$ such that $x_i \cdot y_j = 0$.
\end{lemma}
\begin{proof}
The proof follows from Lemmas~\ref{lemma:matchingGW} and~\ref{lemma:samej_G'}, and from the fact that sub-structures ${G}_W^{(j)}$ are entirely present also in graph $G'$.
\end{proof}

\begin{lemma} \label{lemma:patternsubpattern_G'}
Pattern $P$ has a match in $G'$ if and only if a sub-pattern $\B P_{x_i}\E$ of $P$ has a match in the underlying sub-structure $G_W$ of $G_{WU}$.
\end{lemma}
\begin{proof}
For the $(\Rightarrow)$ implication, notice that the $\E\B$-edges are still present also in graph $G'$. Hence, as in Lemma \ref{lemma:patternsubpattern}, the $\E\B$-edges can only be traversed once in this direction and each distinct sub-pattern $\B P_{x_i}\E$ matches a path from either a distinct portion of $G_{U\ell}$ ($\ell = 1,2$) or $G_{WU}$. Moreover, each occurrence of $P$ must begin with $\B\B$ and end with $\E\E$. String $\B\B$ can be matched only in $G_{WU}$, hence the match must start in this gadget. String $\E\E$ is found either in $G_{U2}$ or in $G_{WU}$. Observe that, by construction, once a match for pattern $P$ is started in $G_{WU}$ the only way in which it can be successfully concluded is either by matching $\E\E$ within this gadget or by first matching a portion of $G_{U2}$ and then \E\E. Because of the structure of the graph, in both cases a sub-pattern $\B P_{x_i}\E$ of $P$ must match one of the sub-structures ${G}_W^{(j)}$ that are present in $G_{WU}$.

The $(\Leftarrow)$ implication is trivial and can be proven as in Lemma~\ref{lemma:patternsubpattern}.
\end{proof}

Thanks to the above lemmas, the correctness of this reduction can be easily proven also for graph $G'$.
\begin{theorem}
Pattern $P$ has a match in $G'$ if and only if there exist vectors $x_i \in X$ and $y_{j} \in Y$ which are orthogonal.
\end{theorem}
\begin{proof}
The proof follows from Lemmas~\ref{lemma:samej_G'}, \ref{lemma:matchingGW'} and~\ref{lemma:patternsubpattern_G'}, just like in the proof of Theorem~\ref{theorem:Emlowerbound}.
\end{proof}

To conclude the proof of Theorem~\ref{cor:deterministic-dag}, we are now left with proving that the complexity of the new reduction remains the same, namely $O(nd)$, and that our graph can be tweaked to achieve maximum sum of indegree and outdegree $3$. The complexity can be checked immediately looking at the structure of $G'$. First, notice that gadget $G_W$ has size $O(nd)$ and it is still present within the new gadget $G_{WU}$, hence $G'$ must have at least $O(nd)$ nodes and edges. On the other hand, obtaining $G'$ by merging together $G_W$ and $G_{U1}$ means that the size of $G'$ cannot be greater than the size of $G$, which is indeed $O(nd)$. We conclude that the size of graph $G'$ is $O(nd)$, thus the reduction can be performed in $O(nd)$ time and space. Degree $3$ can be easily obtained applying the same modification used in our previous work \cite{EGM19}, as the $G_U^{(j)}$ and $G_W^{(j)}$ substructures are basically the same and pairs of dummy nodes can be placed to reduce the degree of the $\zero$ and $\one$-nodes. 

\section{Zig-zag Matching}
The lower bound given for the \pmlg problem can be extended to the special case of an undirected graph with maximum degree $2$, that is, an undirected path. To this end, we need to modify the reduction defining a new alphabet, pattern and graph. 

The original alphabet $\Sigma = \{ \B,\E,\zero,\one \}$ is replaced with $\Sigma' = \{ \B,\E,\AA,\BB,\X,\Y \}$. Characters $\one$ and $\zero$ are encoded in the following manner:
\begin{align*}
    &\one = \oneAB\\
    &\zero = \zeroAB
\end{align*}
When such encoding is applied, character \X will be used as a separator marking the beginning and the end of the old characters. As an example, the sub-pattern
\[  P_{x_i} = \one~\zero~\one \]
will be encoded as
\[  P'_{x_i} = \X~\oneAB~\X~\zeroAB~\X~\oneAB~\X \quad .\]

\subsection{Pattern}
A new pattern $P'$ is built applying the aforementioned encoding to each one of the sub-patterns $P_{x_i}$, thus obtaining a new sub-pattern $P'_{x_i}$. We then concatenate all the sub-patterns $P'_{x_i}$ placing the new character \Y to separate them, instead of \texttt{eb}. Finally, we place characters $\B\Y$ at the beginning of the new pattern, and $\Y\E$ at the end. Here follows an example:
\begin{align*}
    &P = \texttt{bb 100 e b 101 ee}\\
    &\\
    \setcounter{MaxMatrixCols}{20}
    &\begin{matrix}
    ~ & ~ & \one & ~ & \zero & ~ & \zero &\\
    P' = \B &\Y~\X & \oneAB & \X & \zeroAB& \X &\zeroAB& \X\\
    ~ & ~ & \one & ~ & \zero & ~ & \one &\\
    ~ & \Y~\X &\oneAB& \X &\zeroAB& \X &\oneAB& \X& \Y~\E
    \end{matrix}
\end{align*}

Note that for each sub-pattern we are introducing a constant number of new characters, hence the size of the entire pattern $P'$ still is $O(nd)$.

\subsection{Graph}
The same encoding used for the pattern will be applied to the graph. The strategy is to encode $G_W$ in a linear structure concatenating chains of nodes representing each sub-structure $G_W^{(j)}$ one after the other. 

The positions $h$ in which both a $\zero$- and a $\one$-node are present are replaced by a path that can be matched both by $\zero = \zeroAB$ and $\one = \oneAB$. Positions $h$ with only a $\zero$-node and no $\one$-node are encoded instead with a path that can be matched only by $\zero = \zeroAB$ (see Figure~\ref{fig:newsubstructures}). We use $\X$-nodes to separate these paths.  We denote by $LG_W^{(j)}$ (\emph{Linear $G_W^{(j)}$}) this linearized version of $G_W^{(j)}$. Moreover, given sub-structure $G_W^{(j)}$, two new $\Y$-nodes will mark the beginning and the ending of its encoding. Figure~\ref{fig:LG_W1} depicts how to perform such transformation for $G_W^{(j)}$.

\begin{figure}[ht]
    \vspace{20pt}
    (a) \par \noindent
    {
        \centering
        \includegraphics[height=3cm]{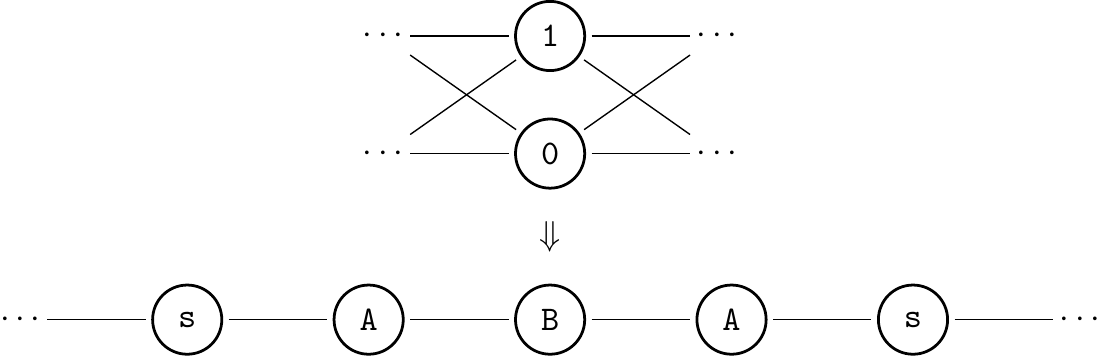}
        \par
    }
    \vspace{25pt}

    (b) \par \noindent
    {   
        \centering
        \includegraphics[height=2.1cm]{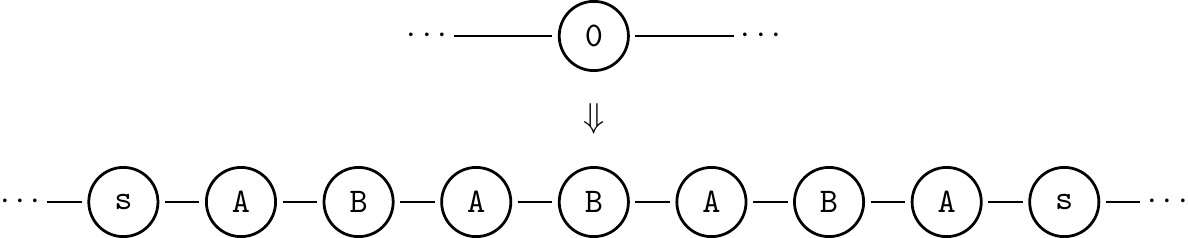}
        \par
    }
    \vspace{20pt}
    \caption{ New substructures. (a) The old substructure is replaced by a chain of nodes that can match either $\X \oneAB \X$ (which represents \one) or $\X\zeroAB\X$ (which represents \zero). (b) The chain of nodes that is replacing a $\zero$-node can match only the string $\X\zeroAB\X$. }
    \label{fig:newsubstructures}
\end{figure}

\begin{figure}[ht]
    \centering
    \vspace{20pt}
    \par \noindent
    \includegraphics[scale=.85]{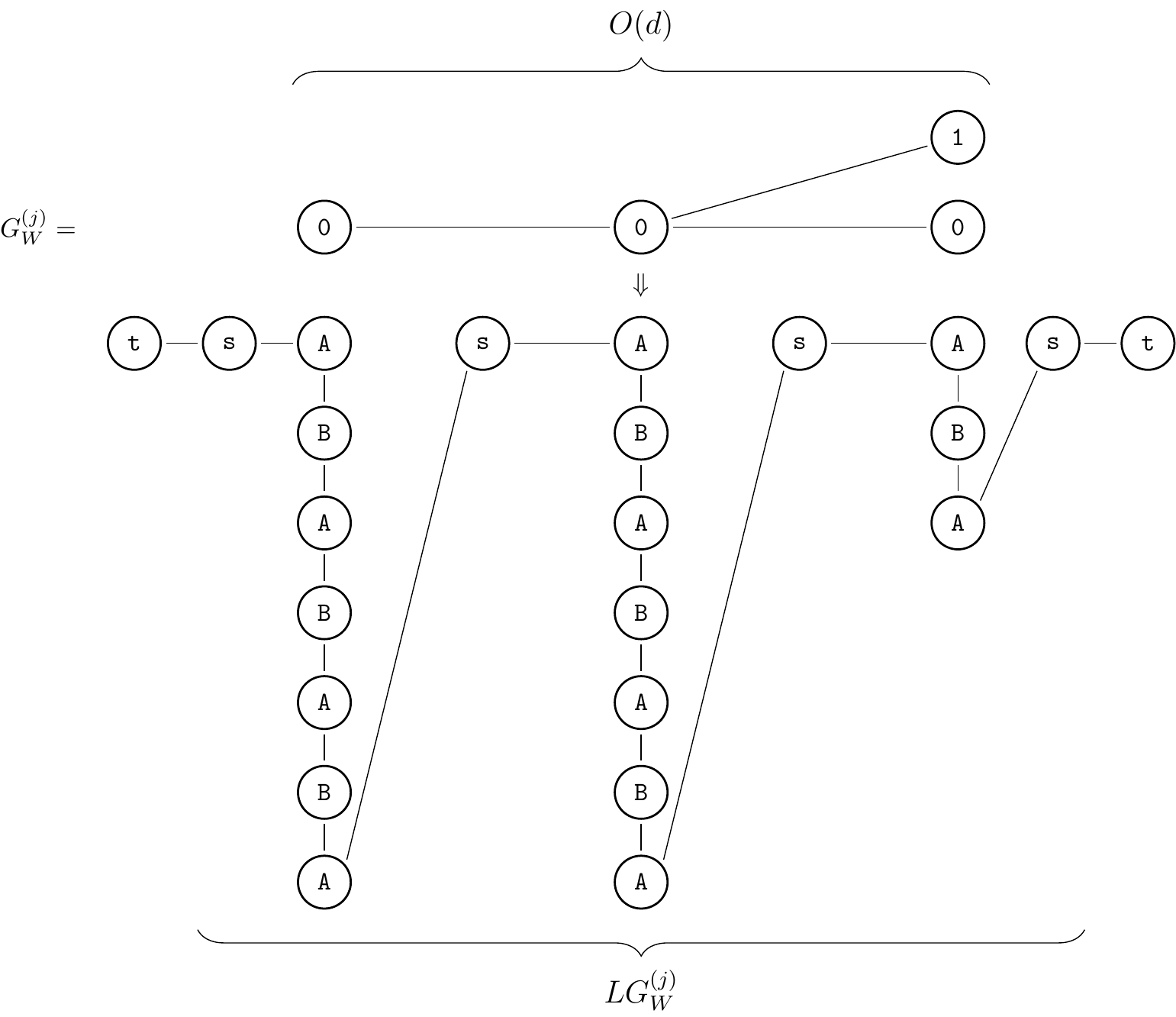}
    \par
    \vspace{20pt}
    \caption{ A row of $G_F$ is converted in a linear structure using characters \X and \Y as markers. }
    \label{fig:LG_W1}
\end{figure}

In a similar manner, $G_U$ is also encoded as a path. We do not need to encode all its $2n-2$ substructures: since the matching path is now allowed to pass through nodes more than once, we only need to encode one of them, in the same manner as done for $G_W^{(j)}$. Let $LG_U$ (\emph{Linear $G_U$}) be the linearized version of a ``jolly'' gadget that was composing the original $G_U$. 

Then, for each $1 \leq j \leq n$, we build structure $LG^{(j)}$ by placing $\Y$-nodes, $LG_U$ instances, $LG_W^{(j)}$, a $\B$-node on the left and an $\E$-node on the right as in Figure~\ref{fig:LG_U-LG_W-LG_U}. 
\begin{figure}[ht]
    \vspace{20pt}
    \centering
    \par \noindent
    \includegraphics[scale=.85]{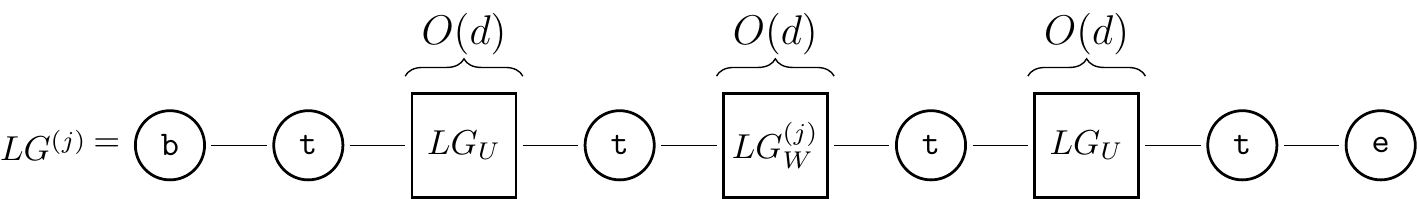}
    \par
    \vspace{20pt}
    \caption{ The $LG_W^{(j)}$ structure surrounded by two instances of $LG_U$. The $\Y$-nodes establish the beginning and the end of a match for a sub-pattern $\Y P'_{x_i}\Y$ while nodes $b$ and $e$ are the starting and ending point for a match of the whole pattern $P'$. }
    \label{fig:LG_U-LG_W-LG_U}
\end{figure} In such structure the $\B$-node and the $\E$-node delimit the beginning and the end of a viable match for a pattern. The $\Y$-nodes are separating the $LG_U$ structures from $LG_W^{(j)}$ and, in general, they are marking the beginning and the end of a match for a sub-pattern $P'_{x_i}$. In order to construct the final graph $LG$ we concatenate all the gadgets $LG^{(1)}, LG^{(2)}, \ldots, LG^{(n)}$ into one chain of nodes. Figure~\ref{fig:LG_Winal} gives a picture of the end result.
\begin{figure}[ht]
    \vspace{20pt}
    \centering
    \par \noindent
    \includegraphics[scale=.85]{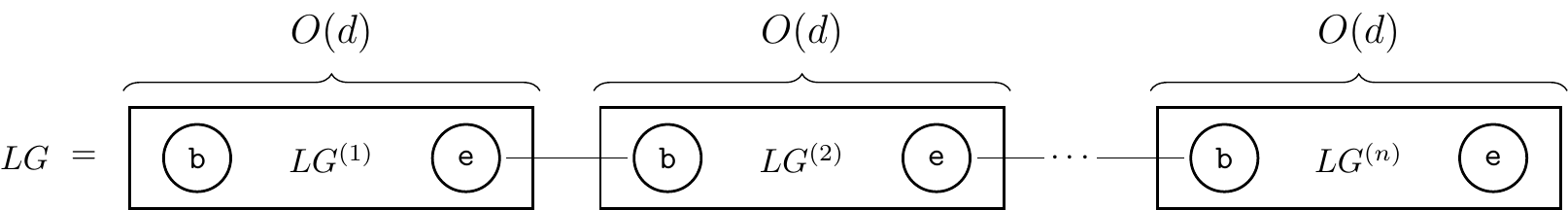}
    \par
    \vspace{20pt}
    \caption{ The final graph $LG$.}
    \label{fig:LG_Winal}
\end{figure} 

No issues arise regarding the size of the graph, since we are replacing every $\zero$-node, or every pair of a $\zero$-node and a $\one$-node, with a constant number of new nodes. By construction, the two gadgets $LG_U$ and $LG_W$ both have size $O(d)$, since for each one of the $d$ entries of a vector we place one of the two possible encodings. In $LG$ there are $n$ instances of $LG_W^{(j)}$, each one surrounded by two $LG_U$ instances. Hence the total size of the graph remains $O(nd)$.
 
\subsection{Putting All Together}
In order to achieve the desired reduction we will prove some properties on $LG$.

We introduce some new notation to simplify the exposition of the upcoming lemmas. We use $t_lLG_W^{(j)}t_r$ to refer to $LG_W^{(j)}$ extended with the $\Y$-nodes on its left and on its right. When referring to the $k$-th \X-character in $P'_{x_i}$ we mean the $k$-th \X-character found scanning $P'_{x_i}$ from left to right; in the same manner we refer to the $k$-th $\X$-node in $LG_W^{(j)}$. Moreover, $\X^{(P'_{x_i})}_k$ denotes the $k$-th \X-character in $P'_{x_i}$, and $s^{ \left( LG_W^{(j)} \right) }_k$ denotes the $k$-th $\X$-node in $LG_W^{(j)}$.

\begin{lemma} \label{lemma:one-to-one_match}
If sub-pattern $\Y P'_{x_i}\Y$ has a match in $t_lLG_W^{(j)}t_r$ starting at $t_l$ and ending at $t_r$, then the $k$-th \X-character in $P'_{x_i}$ matches the $k$-th $\X$-node in $LG_W^{(j)}$, for all $1\leq k \leq d+1$.
\end{lemma}
\begin{proof}
First we prove that all the $\X$-nodes in $t_lLG_W^{(j)}t_r$ are matched exactly once by $\Y P'_{x_i}\Y$. By construction, sub-pattern $P'_{x_i}$ has $d+1$ \X-characters, and $LG_W^{(j)}$ has $d+1$ $\X$-nodes. Since we are working on a chain of nodes and the match is starting at $t_l$ and ending at $t_r$, all the nodes between $t_l$ and $t_r$ have to be matched at least once by $P'_{x_i}$. Assume by contradiction that one such $\X$-node is matched more than once. Sub-pattern $P'_{x_i}$ is left with strictly less than $d$ \X-characters available for matching the other $d$ $\X$-nodes and we reach a contradiction. Now we can prove the statement of the lemma by induction on $k$, i.e the index of the \X-characters and $\X$-nodes.

\textbf{Base Case} $k = 1$. The match starts at $t_l$ hence the only node that $\X^{(P'_{x_i})}_1$ can match is the first $\X$-node to the right on $t_l$, i.e., $s^{ \left( LG_W^{(j)} \right) }_1$.

\textbf{Inductive Case} $k > 1$. The inductive hypothesis tells us that all the nodes up to $s^{ \left( LG_W^{(j)} \right) }_k$ have been matched by consecutive \X-characters of $P'_{x_i}$ up to $\X^{(P'_{x_i})}_k$. We have to prove the statement for $k+1$. Starting from node $s^{ \left( LG_W^{(j)} \right) }_k$ the next $\X$-nodes that can be matched by $\X^{(P'_{x_i})}_{k+1}$ are $s^{ \left( LG_W^{(j)} \right) }_{k-1}$ and $s^{ \left( LG_W^{(j)} \right) }_{k+1}$. Character $\X^{(P'_{x_i})}_{k+1}$ cannot match node $s^{ \left( LG_W^{(j)} \right) }_{k-1}$ since it has already been matched by $s^{(P'_{x_i})}_{k-1}$ and, as argued earlier, every $\X$-node can be matched only once. Thus $\X^{(P'_{x_i})}_{k+1}$ has to match $s^{ \left( LG_W^{(j)} \right) }_{k+1}$.
\end{proof}

\begin{lemma} \label{lemma:matchingLGW}
Sub-pattern $\Y P'_{x_i}\Y$ has a match in $t_lLG_W^{(j)}t_r$ starting at $t_l$ and ending at $t_r$ if and only if there exist $y_j \in Y$ such that $x_i \cdot y_j = 0$.
\end{lemma}
\begin{proof}
This property has already been proved for gadget $G_W$ in Lemma~\ref{lemma:matchingGW}, thus what we are left to prove is that $LG_W^{(j)}$ behaves the same as the sub-gadget $G_W^{(j)}$. First recall that in the construction of $LG_W^{(j)}$ we placed an encoded \one if in $G_W^{(j)}$ we had both a $\zero$-node and a $\one$-node in the same position, while we placed an encoded \zero if we had only a $\zero$-node. Lemma~\ref{lemma:one-to-one_match} guarantees that the encoding in $P'$ of a single character of $P$ is aligned with the encoding in $LG_W^{(j)}$ of a single node of $G_W$, preventing (the encoding of) a character of $P$ from matching (the encoding of) multiple nodes of $G_W$ and vice versa. By construction, $\one = \oneAB$ can match the encoding of a $\one$-node while it fails to match the encoding of the $\zero$-nodes, since their encoding involves too many characters. On the other hand, $\zero = \zeroAB$ can match an encoded $\zero$-node with a natural alignment, but it can also match the encoding of a $\one$-node by scanning it forwards, backwards and forwards again. Therefore the logic behind $LG_W^{(j)}$ safely implements the one of $G_W^{(j)}$, and from this point onward, one can follow the same reasoning as in Lemma~\ref{lemma:matchingGW} to complete the proof.
\end{proof}

The main difference with the original proof resides in assuming that a match for $P'_{x_i}$ starts at $t_l$ and ends at $t_r$. This feature is crucial for the correctness of the reduction and can be safely exploited since, as shown in the following, the $\B$- and $\E$-nodes guarantee that in case of a match for $P'$ we will cross the $LG_W^{(j)}$ gadget from left to right at least once.

\begin{lemma} \label{lemma:patternsubpattern_ZigZag}
Pattern $P'$ has a match in $LG$ if and only if there exist $i$ and $j$ such that $i$ is even and sub-pattern $\Y P'_{x_i}\Y$ has a match in $t_lLG_W^{(j)}t_r$ starting at $t_l$ and ending at $t_r$.
\end{lemma}
\begin{proof}
For the ($\Rightarrow$) implication, first observe that the $\B$- and $\E$-nodes in $LG$ are forcing a direction to follow. Let $LG_{Ul}^{(j)}$ and $LG_{Ur}^{(j)}$ be the $LG_U$ gadgets to the left and to the right of $LG_W^{(j)}$, respectively. Since pattern $P'$ starts with a \B\ and ends with an \E, a match can only start at the $\B$-node on the left of $LG_{Ul}^{(j)}$ and end at the $\E$-node on the right of $LG_{Ur}^{(j)}$, for some $j$. Hence $LG_W^{(j)}$ needs to be crossed by a match from left to right at least once. Thus, there must exist a sub-pattern $\Y P'_{x_i}\Y$ that has a match starting at $t_l$ and ending at $t_r$. For such a pattern Lemma~\ref{lemma:matchingLGW} applies. Moreover, because of our construction, only a sub-pattern on even position can achieve such a match.

The ($\Leftarrow$) implication is immediate since given a sub-pattern $\Y P'_{x_i}\Y$ which has a match in $t_lLG_U^{(j)}t_r$ one can match $\B\Y P'_{x_1}\Y \ldots \Y P'_{x_{i-1}}\Y$ in $LG_{Ul}^{(j)}$  and $\Y P'_{x_{i+1}}\Y \ldots \Y P'_{x_n}\Y\E$ in $LG_{Ur}^{(j)}$ and have a full match for $P'$ in $LG$.
\end{proof}

Since Lemma~\ref{lemma:patternsubpattern_ZigZag} gives us a property which holds only if a sub-pattern is in even position, we need to tweak pattern $P'$ to make the reduction work. Indeed, we define two patterns. The first pattern $P'^{(1)}$ is $P'$ itself; the second pattern $P'^{(2)}$ is obtained by swapping the sub-patterns $P'_{x_i}$ on odd position with the next sub-patterns $P'_{x_{i+1}}$ on even position, for every $i = 1, 3, \ldots$. For example, if $n$ is even, we will have:
\begin{align*}
    P'^{(1)} &= \B\Y~P'_{x_1}~\Y~P'_{x_2}~\Y~P'_{x_3}~\Y~P'_{x_4}~\Y~ \ldots ~\Y~P'_{x_{n-1}}~\Y~ P'_{x_n}~\Y\E = P'\\
    P'^{(2)} &= \B\Y~P'_{x_2}~\Y~P'_{x_1}~\Y~P'_{x_4}~\Y~P'_{x_3}~\Y~ \ldots ~\Y~P'_{x_n}~\Y~ P'_{x_{n-1}}~\Y\E
\end{align*}
In this way, if $n$ is even, we can test every sub-pattern against $LG_W$. While $P'^{(1)}$ checks the even positions of $P'$, $P'^{(2)}$ checks the odd ones. If $n$ is odd then the last sub-pattern would not have the chance to be matched against any $G_W^{(j)}$. In such case we can simply add a dummy sub-pattern $\bar{P} = \X~\oneAB~\X~\oneAB~\X \ldots \X~\oneAB~\X$ (with $d$ repetitions of $\oneAB$) at the end of pattern $P$ as it were the last sub-pattern so that its number of sub-patterns can be even. Indeed, observe that $\bar{P}$ corresponds to vector $\bar{x} = (1 1 \ldots 1)$, which has null product only with vector $\bar{y} = (0 0 \ldots 0)$. Hence if $\bar{y} \not\in Y$ then $\bar{P}$ does not have a match in any $LG^{(j)}$, while if $\bar{y} \in Y$ every sub-pattern $P'_{x_i}$ has a match in the $LG^{(j)}$ built on top of $\bar{y}$. This means that $\bar{P}$ does not disrupt our reduction.

Now we are ready to present the end result.

\begin{theorem}
Either $P'^{(1)}$ or $P'^{(2)}$ has a match in $LG$ if and only if there exist vectors $x_i \in X$ and $y_j \in Y$ which are orthogonal.
\end{theorem}

\begin{proof}
For ($\Rightarrow$) we assume that either $P'^{(1)}$ or $P'^{(2)}$ have a match in $LG$. By Lemma~\ref{lemma:patternsubpattern_ZigZag} this means that there exists a sub-pattern $P'^{(q)}_{x_i}, \; q \in \{1,2\}$ which has a match in $LG_W^{(j)}$, for some $j$. Lemma~\ref{lemma:matchingLGW} then ensures that $x_i \cdot y_j = 0$, thus $x_i$ and $y_j$ are orthogonal. For the other implication ($\Leftarrow$) we assume that there exists two orthogonal vectors $x_i \in X$ and $y_j \in Y$. Thanks to Lemma~\ref{lemma:matchingLGW} we find a sub-pattern $P'_{x_i}$ matching $LG_W^{(j)}$. By construction, $P'_{x_i}$ has to be in even position either in $P'^{(1)}$ or in $P'^{(2)}$. By Lemma~\ref{lemma:patternsubpattern_ZigZag} this means that either $P'^{(1)}$ or $P'^{(2)}$ has a match in $LG$.
\end{proof}

Theorem~\ref{thm:zig-zag} on page \pageref{thm:zig-zag} follows directly from the correctness of these constructions.



\end{document}